\documentclass[12pt,draftclsnofoot,onecolumn,journal,letterpaper]{IEEEtran}
\usepackage[final]{graphicx}
\usepackage[reqno]{amsmath}
\usepackage{amssymb}
\usepackage{url}
\usepackage{color}
\usepackage{times}
\usepackage{wrapfig}
\usepackage{subfig}
\usepackage{amsthm}
\theoremstyle{plain}
\newtheorem{theorem}{Theorem}
\newtheorem{lemma}[theorem]{Lemma}%[chapter]
\newtheorem{remark}[theorem]{Remark}%[chapter]

\definecolor{emph}{rgb}{0.61,0.00,0.00}
\graphicspath{{images/}{../images/}}

\begin{document}

\title{Polar Coding for Fading Channels:\\Binary and Exponential Channel Cases}

\author{
 \IEEEauthorblockN{Hongbo~Si, O.~Ozan~Koyluoglu,~\IEEEmembership{Member,~IEEE,} and Sriram~Vishwanath,~\IEEEmembership{Senior Member,~IEEE}
}
%\thanks{Manuscript received Nov. 20 2013; revised Feb. 18 2014 and May 21 2014; accepted Jun. 25 2014.}
%\thanks{This paper was presented in part at the 2013 Information Theory Workshop, Seville, Spain, Sep. 2013.}
\thanks{H.~Si and S.~Vishwanath are with the Laboratory for Informatics, Networks, and Communications,
Wireless Networking and Communications Group,
The University of Texas at Austin, Austin, TX 78712.
Email: sihongbo@mail.utexas.edu, sriram@austin.utexas.edu.}
\thanks{O.~O.~Koyluoglu is with the
Department of Electrical and Computer Engineering,
The University of Arizona, Tucson, AZ 85721.
Email: ozan@email.arizona.edu.}
%\thanks{Communicated by A.~Graell~i~Amat, Editor for IEEE Transactions on Communications.}
%\thanks{Color versions of one or more of the figures in this paper are available online at http://ieeexplore.ieee.org.}
%\thanks{\copyright~2014 IEEE. Personal use of this material is permitted. However, permission to use this material for any other purposes must be obtained from the IEEE by sending a request to pubs-permissions@ieee.org.}
%\thanks{Digital Object Identifier XX.XXXX/TIT.XXXX.XXXXXX}
}

%\markboth{IEEE TRANSACTIONS ON COMMUNICATIONS, ~VOL.~XX, NO.~XX,~XXXXXX~201X}{Si \MakeLowercase{\textit{et al.}}: Polar Coding for Fading Channels: Binary and Exponential Channel Cases}

%\IEEEpubid{XXXX--XXXX/XX\$XX.XX~\copyright~201X IEEE}

\maketitle

\begin{abstract}

This work presents a polar coding scheme for fading channels, focusing primarily on fading binary symmetric and additive exponential noise channels. For fading binary symmetric channels, a hierarchical coding scheme is presented, utilizing polar coding both over channel uses and over fading blocks. The receiver uses its channel state information (CSI) to distinguish states, thus constructing an overlay erasure channel over the underlying fading channels. By using this scheme, the capacity of a fading binary symmetric channel is achieved without CSI at the transmitter. Noting that a fading AWGN channel with BPSK modulation and demodulation corresponds to a fading binary symmetric channel, this result covers a fairly large set of practically relevant channel settings.

For fading additive exponential noise channels, expansion coding is used in conjunction to polar codes. Expansion coding transforms the continuous-valued channel to multiple (independent) discrete-valued ones. For each level after expansion, the approach described previously for fading binary symmetric channels is used. Both theoretical analysis and numerical results are presented, showing that the proposed coding scheme approaches the capacity in the high SNR regime. Overall, utilizing polar codes in this (hierarchical) fashion enables coding without CSI at the transmitter, while approaching the capacity with low complexity.
\end{abstract}


\IEEEpeerreviewmaketitle

\section{Introduction}

Polar codes are the first family of provably capacity achieving codes for symmetric binary-input discrete memoryless channels (B-DMC) with low encoding and decoding complexity \cite{Arikan:Channel08} \cite{Arikan:Error09}. These codes polarize the underlying channel in the sense that, via channel combining and channel splitting stages, multiple uses of the given channel are transformed into equivalent polarized ones: either purely noisy (referred to as ``bad'' channel instances) or noiseless (referred to as ``good'' channel instances). Then, information symbols are mapped to the good instances of polarized channels, whereas channel inputs corresponding to the bad instances are fixed and shared between the transmitter and receiver. It is shown in \cite{Arikan:Channel08} that the fraction of the good channel instances approaches the symmetric capacity of the channel, which is equal to the capacity of the underlying channel if the channel is symmetric. That is, polar codes achieve the capacity of symmetric B-DMCs. This phenomenon of channel polarization has then been generalized to arbitrary discrete memoryless channels with a construction complexity to the same order and a similar error probability behavior \cite{Sasoglu:Polarization09}. Moreover, polar codes are proved to be optimal for lossy compression with respect to binary symmetric source \cite{Arikan:Source10}\cite{Korada:Source10}, and then further extended to coding for larger source alphabets \cite{Mohammad:Source10}.

\IEEEpubidadjcol

Recently, polar codes have been adapted to channels with non-discrete inputs as well. In \cite{Abbe:Polar11}, using polarization results for multiple access channels \cite{Abbe:Polar12}, a polar coding scheme for additive white Gaussian noise (AWGN) channels is presented. It is shown that the approach of using multiple access channel coding with a large number of binary-input users possesses much better complexity attributes compared to that of using single-user channels with large input cardinality. In a separate work \cite{Ozan:Expansion12}, by adopting discrete polar codes as embedded codes, an expansion coding approach is presented, and the capacity of additive exponential noise channel is shown to be achievable in the high SNR regime.

The analysis of polar coding for fading channels, with either discrete-valued or continuous-valued noises, is still limited. Recent work \cite{Boutros:Fading13} investigates binary input real number output AWGN fading channel, where the fading coefficient is assumed to be one of the two states with equal probabilities. These fading coefficients are assumed to follow arbitrary distributions with the requirement of satisfying some tail probability constraints. For this setup, the authors proposed polar coding schemes where symbols are multiplexed in a specific fashion at the encoder. In particular, the paper analyzes diagonal, horizontal, and uniform multiplexers; and, the corresponding diversity and outage analysis have been performed. Another recent work \cite{Bravo:Fading13} focuses on polar coding schemes for Rayleigh fading channel under two scenarios: block fading with known channel state information (CSI) at the transmitter and fast fading with fading distribution known at the transmitter. For the latter case, the channel is shown to be symmetric, and through quantization of the channel output, the polar coding scheme is shown to achieve a constant gap to the capacity.

In this work, we focus on a block fading model without the CSI at the transmitter, and propose a hierarchical polar coding scheme for such channels. First, we focus on fading binary symmetric channel (BSC), which is an important model as it is closely related to an AWGN block fading channel with BPSK modulation and demodulation. Such binary input AWGN models are previously analyzed in \cite{Arikan:Polar11}\cite{Tse:Polar12} to evaluate the performance of polar codes over AWGN channels. Here, we focus on communication channel models that involve fading, where the channel coefficients vary according to a block fading model. This scenario of fading AWGN with BPSK modulation resembles a fading binary symmetric channel model, where each fading block has a cross-over probability depending on the corresponding channel state realization. Specifically, AWGN channel states with higher SNRs map to binary symmetric channels with lower crossover probabilities. For this binary symmetric fading model, we propose a novel polar coding approach that utilizes polarization in a \emph{hierarchical} manner \emph{without} channel state information (CSI) at the transmitter (with channel state statistics assumed to be known at the transmitter). The key factor enabling our coding scheme is the hierarchical utilization of polar coding. More precisely, polar codes are not only designed over channel uses for each fading state, but also utilized over fading blocks. By taking advantage of the degradedness property of channel polarization between different BSCs, an erasure model (over fading blocks) is constructed for every channel instance that polarizes differently depending on the channel states. It is shown that this proposed coding scheme, without instantaneous CSI at the transmitter, achieves the capacity of the fading binary symmetric channel.

As an additional analog fading model (in addition to the AWGN fading scenario with BPSK modulation), we consider additive exponential noise channels.
In analog channels, the additive exponential noise (AEN) channel is of particular interest as it models worst-case noise given mean and non-negativity constraints \cite{Verdu:Exponential96}. In addition, the AEN model naturally arises in non-coherent communication settings, and in optical communication scenarios \cite{Martinez:Communication11}\cite{LeGoff:Capacity11}. In \cite{Ozan:Expansion12}, an expansion scheme is proposed to achieve the capacity of AEN channels. Here, we adopt a similar approach for the fading AEN channels. In particular, due to the decomposition property of the exponential distribution, we show that a fading AEN channel can be transformed into a set of fading BSCs. Then, by employing the aforementioned polar coding scheme for fading BSCs at each level, we show that the proposed method approaches to the capacity of fading AEN channels in the high SNR regime.

In both cases considered in this paper, utilizing polar codes in such a novel (hierarchical) way enables coding without CSI at the transmitter, a practically important scenario in wireless systems. In addition, the low encoding and decoding complexity of polar codes are inherited in the proposed schemes. (The scaling of complexity with respect to the system parameters are detailed in the later parts of the sequel.) Therefore, the proposed approach, by having both low complexity and realistic CSI assumption properties, is suitable for practical utilization of polar codes over fading channels (especially for those channels with long fading coherence intervals).

The rest of paper is organized as follows. After a brief introduction of the preliminary results on polar codes in Section~II, the polar coding scheme for fading binary symmetric channels is detailed in Section~III. Section~IV is devoted to the study of fading additive exponential noise channels. Finally, concluding remarks are provided in Section~V.

%%%%%%%%%%%%%%%%%%%%%%%%%%%%%%%%%%%%%%%%%%%%%%%%%%%%%%%%%%%%%%%%%%%%%%%%%%%%%%
%%%%%%%%%%%%%%%%%%%%%%%%%%%%%%%%%%%%%%%%%%%%%%%%%%%%%%%%%%%%%%%%%%%%%%%%%%%%%%

\section{Introduction to Polar Codes}

The construction of polar code is based on the observation of channel polarization. Consider a binary-input discrete memoryless channel $W:\mathcal{X}\to\mathcal{Y}$, where $\mathcal{X}=\{0,1\}$. Define
\begin{equation}
F=\left[\begin{array}{cc}
  1 & 0 \\
  1 & 1
\end{array}\right].\nonumber
\end{equation}
Let $B_N$ be the bit-reversal operator defined in \cite{Arikan:Channel08}, where $N=2^n$. By applying the transform $G_N=B_NF^{\otimes n}$ ($F^{\otimes n}$ denotes the $n^{\text{th}}$ Kronecker power of $F$) to the message $u_{1:N}$, the encoded $x_{1:N}=u_{1:N}G_N$ is transmitted through $N$ independent copies of channel $W$. Then, consider $N$ binary-input coordinate channels $W_N^{(i)}:\mathcal{X}\to\mathcal{Y}^N\times\mathcal{X}^{i-1}$, where for each $i\in\{1,\ldots,N\}$ the transition probability is given by
\begin{equation}
W_N^{(i)}(y_{1:N},u_{1:{i-1}}|u_i)\triangleq \sum_{u_{{i+1}:N}}\frac{1}{2^{N-1}}W^N(y_{1:N}|u_{1:N}G_N).\nonumber
\end{equation}
Here, as $N$ tends to infinity, the channels $\{W_N^{(i)}\}_{1:N}$ polarize to either noiseless or purely noisy ones, and the fraction of noiseless channels is close to $I(W)$, the symmetric capacity of the channel $W$ \cite{Arikan:Channel08}.

To this end, polar codes can be considered as $G_N$-coset codes with parameter $(N,K,\mathcal{A},u_{\mathcal{A}^c})$, where $u_{\mathcal{A}^c}\in\mathcal{X}^{N-K}$ is the frozen vector (can be set to all-zeros for symmetric channels \cite{Arikan:Channel08}, which is the focus of this paper), and the information set $\mathcal{A}$ is chosen as a $K$-element subset of $\{1,\ldots,N\}$ such that the Bhattacharyya parameters satisfy $Z(W_N^{(i)})\leq Z(W_N^{(j)})$ for all $i\in\mathcal{A}$ and $j\in\mathcal{A}^c$. (The indices in $\mathcal{A}$ are ``good'' channel indices, whereas those in $\mathcal{A}^c$ corresponds to ``bad'' channel indices.)

The decoder in the polar coding scheme is a successive cancellation (SC) decoder, which gives an estimate $\hat{u}_{1:N}$ of $u_{1:N}$ given the knowledge of $\mathcal{A}$, $u_{\mathcal{A}^c}$, and $y_{1:N}$ by computing
\begin{align}
\hat{u}_i \triangleq \left\{
\begin{array}{cl}
  0, & \text{if }i\in\mathcal{A}^c, \\
  d_i(y_{1:N},\hat{u}_{1:{i-1}}), & \text{if }i\in \mathcal{A},
\end{array}
\right.\nonumber
\end{align}
in the order $i$ from $1$ to $N$, where
\begin{align}
d_i(y_{1:N},\hat{u}_{1:{i-1}})\triangleq \left\{
\begin{array}{cl}
  0, & \text{if }\frac{W_N^{(i)}(y_{1:N},\hat{u}_{1:{i-1}}|0)}{W_N^{(i)}(y_{1:N},\hat{u}_{1:{i-1}}|1)}\geq 1, \\
  1,& \text{otherwise.}
\end{array}
\right.\nonumber
\end{align}
\cite{Arikan:Channel08} proved that by adopting an SC decoder, polar coding achieves any rate $R<I(W)$ with an error scaling as $O(2^{-N^{\beta}})$, where $\beta<1/2$. Moreover, the encoding and decoding complexity of polar coding are both $O(N\log N)$, where $N$ is the length of codeword.

To summarize, polar codes have excellent rate and complexity, but the code design is sensitive to the channel state information at the transmitter. For example, the choices of ``good'' indices depend on the crossover probability $p$ of a BSC, i.e., the set $\mathcal{A}(p)$ is a function of the value of $p$. Thus, it is not straight-forward to design capacity-achieving polar coding schemes if there is uncertainty in the channel parameters. In particular, for a block fading scenario, if the channel states are not known a priori at the transmitter, one may not know which indices should correspond to ``good'' channel instances. As detailed in the next section, the proposed scheme solves this problem via a hierarchical code design.

%%%%%%%%%%%%%%%%%%%%%%%%%%%%%%%%%%%%%%%%%%%%%%%%%%%%%%%%%%%%%%%%%%%%%%%%%%%%%%
%%%%%%%%%%%%%%%%%%%%%%%%%%%%%%%%%%%%%%%%%%%%%%%%%%%%%%%%%%%%%%%%%%%%%%%%%%%%%%

\section{Polar Coding for Fading Binary Symmetric Channel}
\label{sec:FadingBSC}
\subsection{System Model}

Fading channels characterize the wireless communication channels, where the channel states vary over channel uses. Fading coefficients typically vary much slower than transmission symbol duration in practice. For such cases, a block fading model is considered, wherein the channel state is assumed to be a constant over each coherence time interval, and follows a stationary ergodic process across fading blocks. For such a block fading model, we consider the practical scenario where the channel state information is available only at the decoder (CSI-D)~\cite[pages 186-187]{Tse:Wireless05}, while the transmitter is assumed to know only the statistics of the channel states.

Binary symmetric channel (BSC) is a channel with binary input $X$, binary noise $Z$, and a binary output $Y=X\oplus Z$. Here, for the fading BSC, the channel noise is a Bernoulli distributed random variable, where its statistics depend on the channel states. For the block fading BSC considered in this work, the channel is modeled as follows.
\begin{equation}
Y_{b,i}=X_{b,i}\oplus Z_{b,i},\quad b=1,\ldots,B,\quad i=1,\ldots,N,\label{fun:channel_definition}
\end{equation}
where $N$ is the block length, and $B$ is the number of fading blocks. Here, $Z_{b,i}$ are assumed to be identically distributed within a block and follow an i.i.d. fading process over blocks. That is,
if we consider fading BSC with $S$ states, with probability $q_s$ the parameter $p_s$ is chosen for the fading block $b$, where the channel noise $Z_{b,i}$ is sampled from a Bernoulli random variable with parameter $p_s$ for all $i\in\{1,\ldots,N\}$. Here, $1\leq s\leq S$ and $\sum\limits_{s=1}^S q_s=1$.

In wireless communications, the fading binary symmetric channel is utilized to model a fading AWGN channel with BPSK modulation and demodulation. In particular, for a fading AWGN channel with input power constraint $P_X$, the channel noise is distributed as i.i.d. Gaussian with variance $P_Z$, and the channel gain (the factor $h$ in the AWGN channel $Y=hX+Z$) remains the same statistic within a fading block, and follows an ergodic process over different blocks.
After utilizing the BPSK modulation and demodulation at the encoder and decoder, respectively, the equivalent channel is a binary input and binary output channel, with transition probability relating to AWGN channel state.
More precisely, if the channel gain for a particular fading block $b$, $h_{b,i}\;\forall\; i$, is equal to $h_s$ with probability $q_s$ for some $s\in\{1,2,\ldots,S\}$, then the corresponding binary noise in the equivalent fading BSC has the statistics of
\begin{equation}
p_s\triangleq\text{Pr}\{Z_{b,i}=1\}=1-\Phi(h_s\sqrt{\text{SNR}}),\label{fun:p_1}
\end{equation}
where $\Phi(\cdot)$ is the CDF of normal distribution and $\text{SNR}$ is the signal-to-noise ratio, i.e. $\text{SNR}=P_X/P_Z$. In other words, the channel at each fading block can be modeled as $W_s\triangleq$BSC$(p_s)$ with probability $q_s$.

The ergodic capacity of a fading binary symmetric channel is given by \cite[pages 584-586]{Gamal:Network11}
\begin{equation}
C_{\text{CSI-D}}=\sum_{s=1}^S q_s[1-H(p_s)],\label{fun:fading_capacity}
\end{equation}
where $H(\cdot)$ is the binary entropy function, and CSI-D refers to channel state information at the decoder. Note that, the ergodic capacity of fading BSC is an average over the capacities of all possible channels corresponding to different channel states. In this section, we propose a polar coding scheme that achieves the capacity of this fading BSC with low encoding and decoding complexity, without having instantaneous channel state information at the transmitter. Towards this end, we first focus on a fading BSC with two channel states, and then generalize our results to arbitrary finite number of channel states.

\subsection{Intuition}

In polar coding for a BSC, we see that the channel can be polarized by transforming a set of independent copies of given channels into a new set of channels whose symmetric capacities tend to $0$ or $1$ (for all but a vanishing fraction of indices). Towards applying such a polarization phenomenon to fading BSC, we first focus on how two binary symmetric channels polarize at the same time. We summarize a result given in \cite{Korada:Thesis09} regarding the polarization of degraded channels.

\begin{lemma}[\cite{Korada:Thesis09}]
For two binary symmetric channels $W_1\triangleq\text{BSC}(p_1)$ and $W_2\triangleq\text{BSC}(p_2)$, if $W_1$ is degraded with respect to $W_2$, i.e. $p_1\geq p_2$, then for any channel index $i\in\{1,\ldots,N\}$, the reconstructed channels after polarization have the relationship that $W_{1,N}^{(i)}$ is degraded with respect to $W_{2,N}^{(i)}$, and hence $I(W_{1,N}^{(i)})\leq I(W_{2,N}^{(i)})$.
\end{lemma}

That is, when polarizing two binary symmetric channels, the reconstructed channels of the degraded channel have lower symmetric rate compared to that of the other channel. This statement also implies that
$$
\mathcal{A}_1\subseteq \mathcal{A}_2,% \label{fun:channel_compare}
$$
where $\mathcal{A}_1$ and $\mathcal{A}_2$ denote the information sets of the degraded and superior channels, respectively. This relationship is illustrated in Fig.~\ref{fig:Polarization}. Based on this observation, when polarizing two BSCs, the channel indices after reordering permutation $\pi$ can be divided into three categories (we assume that channel $1$ is degraded, i.e., $p_1\geq p_2$):
\begin{enumerate}
\item Set $\mathcal{G}$: both channels are good, i.e.,
$$I(W_{1,N}^{(\pi(i))})\to 1,\quad I(W_{2,N}^{(\pi(i))})\to 1.$$
\item Set $\mathcal{M}$: only channel $2$ is good, while channel $1$ is bad, i.e.,
$$I(W_{1,N}^{(\pi(i))})\to 0,\quad I(W_{2,N}^{(\pi(i))})\to 1.$$
\item Set $\mathcal{B}$: both channels are bad, i.e.,
$$I(W_{1,N}^{(\pi(i))})\to 0,\quad I(W_{2,N}^{(\pi(i))})\to 0.$$
\end{enumerate}

\begin{figure}[t!]
 \centering
 \includegraphics[width=0.8\columnwidth]{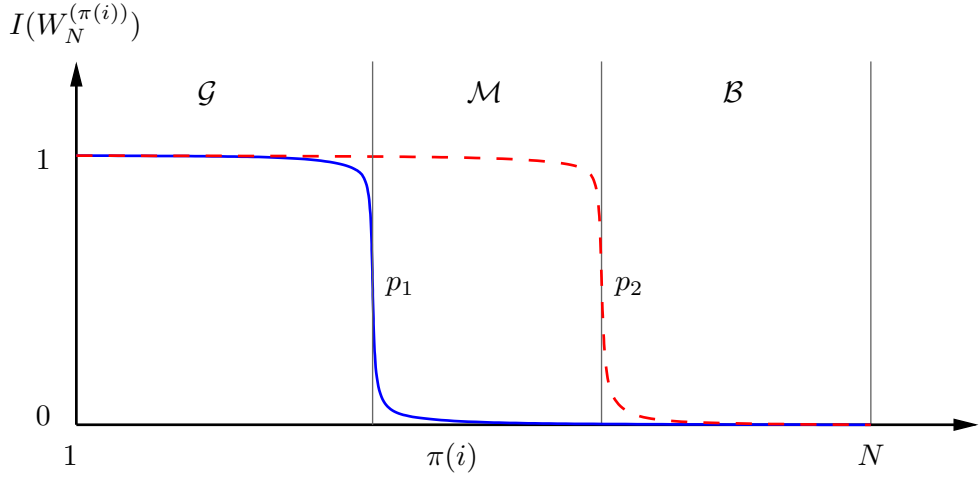}
 \caption{{\bf Illustration of polarizations for two binary symmetric channels.} The blue-solid line represents the degraded channel with transition probability $p_1$, and the red-dashed one represents the superior channel with $p_2$ ($p_1\geq p_2$). Values of $I(W_N^{(\pi(i))})$, the reordered symmetric mutual information, are shown for both channels.
}
\label{fig:Polarization}
\end{figure}

We have the following relationships between these sets. First, information sets for two channels are given by  $\mathcal{A}_1=\mathcal{G}$, and $\mathcal{A}_2=\mathcal{G}\cup \mathcal{M}$. Moreover, considering the sizes of these sets, we have:
\begin{align}
&|\mathcal{G}|=|\mathcal{A}_1|=N[1-H(p_1)-\epsilon],\label{fun:size_G}\\
&|\mathcal{M}|=|\mathcal{A}_2|-|\mathcal{A}_1|=N[H(p_1)-H(p_2)],\label{fun:size_M}\\
&|\mathcal{B}|=N-|\mathcal{A}_2|=N[H(p_2)+\epsilon],\label{fun:size_B}
\end{align}
where $\epsilon$ is an arbitrarily small positive number (that vanishes as $N\to\infty$).

For a fading binary symmetric channel, we again utilize Fig.~\ref{fig:Polarization} to illustrate our coding scheme. Here, consider a fading BSC with only two fading states, the degraded state and the superior one (denoted as state $1$ and $2$, respectively). If channel is in state $1$, which happens with probability $q_1$, the fading channel polarizes to the blue-solid curve, and otherwise the channel is in state $2$, which happens with the probability $q_2=1-q_1$, and the fading channel polarizes to the red-dashed curve. Hence, the reconstructed channel with index in set $\mathcal{G}$ always polarizes to a good one, i.e., its symmetric mutual information is close to $1$ no matter what the fading state is. And, the reconstructed channel with index in set $\mathcal{B}$ always polarizes to a bad one, i.e., its symmetric mutual information is close to $0$ no matter what the fading state is. Therefore, one can reliably transmit information for channel instances belonging to $\mathcal{G}$, whereas one may not transmit any information for channel instances belonging to $\mathcal{B}$. The novel part of the proposed coding scheme is for the middle region, i.e., coding over the set $\mathcal{M}$, where reconstructed channels polarize differently depending on the channel states. Since we consider the transmitter has no prior knowledge of channel states before transmitting, coding over channels with indices in $\mathcal{M}$ is challenging. At this point, we observe that for these channels, with probability $q_2$ they are nearly noiseless, and with probability $q_1$ they are purely noisy. Thus, each channel can be modeled as a binary erasure channel (BEC) from the viewpoint of blocks, where the erasure probability is equal to $q_1$. Here, we denote this channel as
$$\tilde{W}\triangleq\text{BEC}(q_1).$$
This observation motivates our design of hierarchical encoder and decoder for fading BSCs.

\subsection{Encoder}

The encoding process has two phases, hierarchically using polar codes to generate $NB$-length codewords, where $N$ is blocklength and $B$ is the number of blocks.
\subsubsection{Phase 1 (BEC Encoding)} In this phase, we generate $|\mathcal{M}|$ number of BEC polar codes, each with length $B$.
Consider a set of blockwise messages $v^{(k)}$ with $k\in\{1,\ldots,|\mathcal{M}|\}$. For every $v^{(k)}$, construct polar codeword $\tilde{u}^{(k)}$, which is formed by the $G_B$-coset code with parameter $(B,|\tilde{\mathcal{A}}|,\tilde{\mathcal{A}},0)$, where $\tilde{\mathcal{A}}$ is the information set for $\tilde{W}=\text{BEC}(q_1)$, and we choose the rate to be optimal, i.e.
\begin{equation}
|\tilde{\mathcal{A}}|=(1-q_1-\epsilon)B.\label{fun:size_A}
\end{equation}
In other words, we construct a set of polar codes, where each code corresponds to an index in set $\mathcal{M}$, with the same rate $1-q_1-\epsilon$, the same information set $\tilde{\mathcal{A}}$, and the same frozen values, $0$. More precisely, if we denote the reordering permutation for $\tilde{W}$ as $\tilde{\pi}$, then
    \begin{align}
    &\tilde{\pi}(v^{(k)})=[v^{(k)}_1,\ldots,v^{(k)}_{|\tilde{\mathcal{A}}|},0,\ldots,0],\label{fun:code_phase1}\\
    &\tilde{u}^{(k)}=v^{(k)}G_B.\label{fun:encoding_phase1}
    \end{align}

\subsubsection{Phase 2 (BSC Encoding)} In this phase, we generate $B$ number of BSC polar codes, each with length $N$.
Consider a set of messages $u^{(b)}$ with $b\in\{1,\ldots, B\}$. For every $u^{(b)}$, construct polar codeword $x^{(b)}$, which is $G_N$-coset code with parameter $(N,|\mathcal{G}|,\mathcal{G},u^{(b)}_{\mathcal{G}^c})$, where $\mathcal{G}$ is BSC information set with size given by (\ref{fun:size_G}). Remarkably, we do not set all non-information bits to be $0$, but transpose the blockwise codewords generated from Phase 1 and embed them into the messages of this phase. More precisely, if denote the reordering permutation operator of BSC as $\pi$, then
    \begin{align}
    &\pi(u^{(b)})=[u^{(b)}_1,\ldots,u^{(b)}_{|\mathcal{G}|},\tilde{u}^{(1)}_{b},\ldots,\tilde{u}^{(|\mathcal{M}|)}_{b},0,\ldots,0],\label{fun:code_phase2}\\
    &x^{(b)}=u^{(b)}G_N.\label{fun:encoding_phase2}
    \end{align}
By collecting all $\{x^{(b)}\}_{1:B}$ together, the encoder outputs a codeword with length $NB$. We equivalently express these codewords by a $B\times N$ matrix. The proposed encoder for fading binary symmetric channel is illustrated in Fig.~\ref{fig:Encoder}.

\begin{figure*}[t]
 \centering
 \includegraphics[scale=0.7]{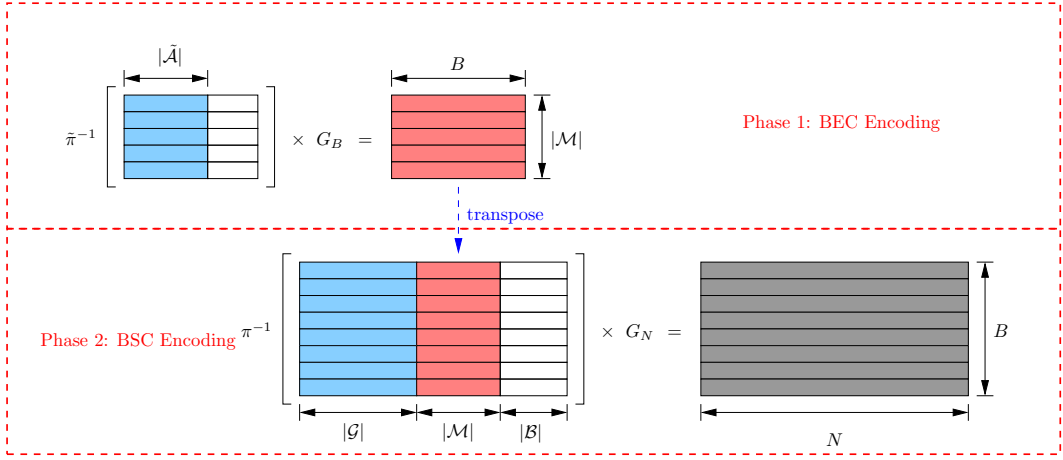}
 \caption{{\bf Illustration of proposed polar encoder for a fading binary symmetric channel with two states.} Bits in blue are information bits, and those in white are frozen as zeros. The codewords generated from Phase 1 are transposed and embedded into the messages of Phase 2 to generate the final codeword of length $NB$. $\tilde{\pi}$ and $\pi$ are column reordering permutations with respect to BEC and BSC, correspondingly.
}
\label{fig:Encoder}
\end{figure*}

\subsection{Decoder}

After receiving the sequence $y_{1:{NB}}$ from the channel, the decoder's task is to make estimates $\{\hat{v}^{(k)}\}_{1:{|\mathcal{M}|}}$ and $\{\hat{u}^{(b)}\}_{1:{B}}$, such that the information bits in both sets of messages match the ones at the transmitter with high probability. Rewrite channel output $y_{1:{NB}}$ as a $B\times N$ matrix, with row vectors $\{y^{(b)}\}_{1:B}$. As that of the encoding process, the decoding process also works in phases:

\subsubsection{Phase 1 (BSC Decoding I)} In this phase, we decode part of the output blocks using the BSC SC decoder with respect to the superior channel state.
More precisely, since the receiver knows channel states, it can adopt the correct SC decoder (BSC$(p_2)$ SC decoder in this case) to obtain $\hat{u}^{(b)}$ from $y^{(b)}$ for every $b$ corresponding to the superior channel state. To this end, the ``\emph{BSC Decoder I}'' (for block $b$ with the superior fading state) in this phase is described as follow:
    \begin{align}
\hat{u}^{(b)}_i \triangleq \left\{
\begin{array}{ll}
  0, & \text{if }i\in\mathcal{B}, \\
  d_{2,i}(y^{(b)}_{1:N},\hat{u}^{(b)}_{1:i-1}), & \text{if }i\in \mathcal{G}\cup\mathcal{M},
\end{array}
\right.\nonumber
\end{align}
\indent in the order $i$ from $1$ to $N$, where
\begin{align}
d_{2,i}(y^{(b)}_{1:N},\hat{u}^{(b)}_{1:i-1})\triangleq \left\{
\begin{array}{ll}
  0, & \text{if }\frac{W_{2,N}^{(i)}(y^{(b)}_{1:N},\hat{u}^{(b)}_{1:i-1}|0)}{W_{2,N}^{(i)}(y^{(b)}_{1:N},\hat{u}^{(b)}_{1:i-1}|1)}\geq 1, \\
  1,& \text{otherwise.}
\end{array}
\right.\nonumber
\end{align}

In this phase, one can reliably decode the information bits in blocks with respective to the superior channel states (with the knowledge of frozen symbols corresponding to $\mathcal{B}$ indices). For the blocks with the degraded channel states, information bits cannot be decoded reliably because frozen bits corresponding to set $\mathcal{M}$ is not known at the decoder. At this point, we use the next phase to decode these bits using a BEC SC decoder. For that, a $B\times |\mathcal{M}|$ matrix, $\hat{\tilde{\bold{U}}}$, is constructed by choosing each row as $\hat{u}^{(b)}$ with respective to bits in set $\mathcal{M}$ for the superior channel states. The other elements in this matrix, i.e., the symbols corresponding to the degraded channel states, are set to erasures.

\subsubsection{Phase 2 (BEC Decoding)} In this phase, we decode the frozen bits with respect to the degraded states using BEC SC decoders. More precisely, each column of matrix $\hat{\tilde{\bold{U}}}$, denoted by $\hat{\tilde{u}}^{(k)}$ for $k\in\{1,\ldots,|\mathcal{M}|\}$, is considered as the input to the decoder, and the receiver aims to decode $\hat{v}^{(k)}$ from $\hat{\tilde{u}}^{(k)}$ using the SC decoder with respect to channel $\tilde{W}=\text{BEC}(q_1)$. More formally, the ``\emph{BEC Decoder}'' in this phase is expressed by the following:

\begin{align}
\hat{v}^{(k)}_j \triangleq \left\{
\begin{array}{ll}
  0, & \text{if }j\in\tilde{\mathcal{A}}^c, \\
  \tilde{d}_j(\hat{\tilde{u}}^{(k)}_{1:|\mathcal{M}|},\hat{v}^{(k)}_{1:j-1}), & \text{if }j\in \tilde{\mathcal{A}},
\end{array}
\right.\nonumber
\end{align}
for $j$ from $1$ to $B$, where
\begin{align}
\tilde{d}_j(\hat{\tilde{u}}^{(k)}_{1:|\mathcal{M}|},\hat{v}^{(k)}_{1:j-1})\triangleq \left\{
\begin{array}{cl}
  0, & \text{if }\frac{\tilde{W}_{N}^{(j)}(\hat{\tilde{u}}^{(k)}_{1:|\mathcal{M}|},\hat{v}^{(k)}_{1:j-1}|0)}{ \tilde{W}_{N}^{(j)}(\hat{\tilde{u}}^{(k)}_{1:|\mathcal{M}|},\hat{v}^{(k)}_{1:j-1}|1)}\geq 1, \\
  1,& \text{otherwise.}
\end{array}
\right.\nonumber
\end{align}
After Phase 2, the decoder outputs an $|\mathcal{M}|\times B$ matrix $\hat{\bold{V}}$ with rows $\{\hat{v}^{(k)}\}_{1:|\mathcal{M}|}$. Moreover, the decoder can reconstruct all bits erased in matrix $\hat{\tilde{\bold{U}}}$, which is denoted as $\tilde{u}^{(b)}$ for each $b$ corresponding to the degraded channel state. Using this information, we are able to decode the information bits in blocks with the degraded channel state in the next phase.

\subsubsection{Phase 3 (BSC Decoding II)} In this phase, we decode the remaining blocks from Phase 1, using BSC SC decoders with respect to the degraded channel states. In particular, for each block $b$ in the degraded channel state, the receiver could decode $\hat{u}^{(b)}$ from $y^{(b)}$ using BSC$(p_1)$ SC decoder by setting frozen bit as $\tilde{u}^{(b)}_i$ for each $i\in\mathcal{M}$ and $0$ for each $i\in\mathcal{B}$. More formally, we have the ``\emph{BSC Decoder II}'' (for block $b$ with a degraded fading state) described as:
\begin{align}
\hat{u}^{(b)}_i \triangleq \left\{
\begin{array}{ll}
  0, & \text{if }i\in\mathcal{B},\\
  \tilde{u}^{(b)}_i, &\text{if }i\in\mathcal{M},\\
  d_{1,i}(y^{(b)}_{1:N},\hat{u}^{(b)}_{1:i-1}), & \text{if }i\in \mathcal{G},
\end{array}
\right.\nonumber
\end{align}
\indent in the order $i$ from $1$ to $N$, where
\begin{align}
d_{1,i}(y^{(b)}_{1:N},\hat{u}^{(b)}_{1:i-1})\triangleq \left\{
\begin{array}{ll}
  0, & \text{if }\frac{W_{1,N}^{(i)}(y^{(b)}_{1:N},\hat{u}^{(b)}_{1:i-1}|0)}{W_{1,N}^{(i)}(y^{(b)}_{1:N},\hat{u}^{(b)}_{1:i-1}|1)}\geq 1, \\
  1,& \text{otherwise.}
\end{array}
\right.\nonumber
\end{align}

\begin{figure*}[t]
 \centering
 \includegraphics[scale=0.7]{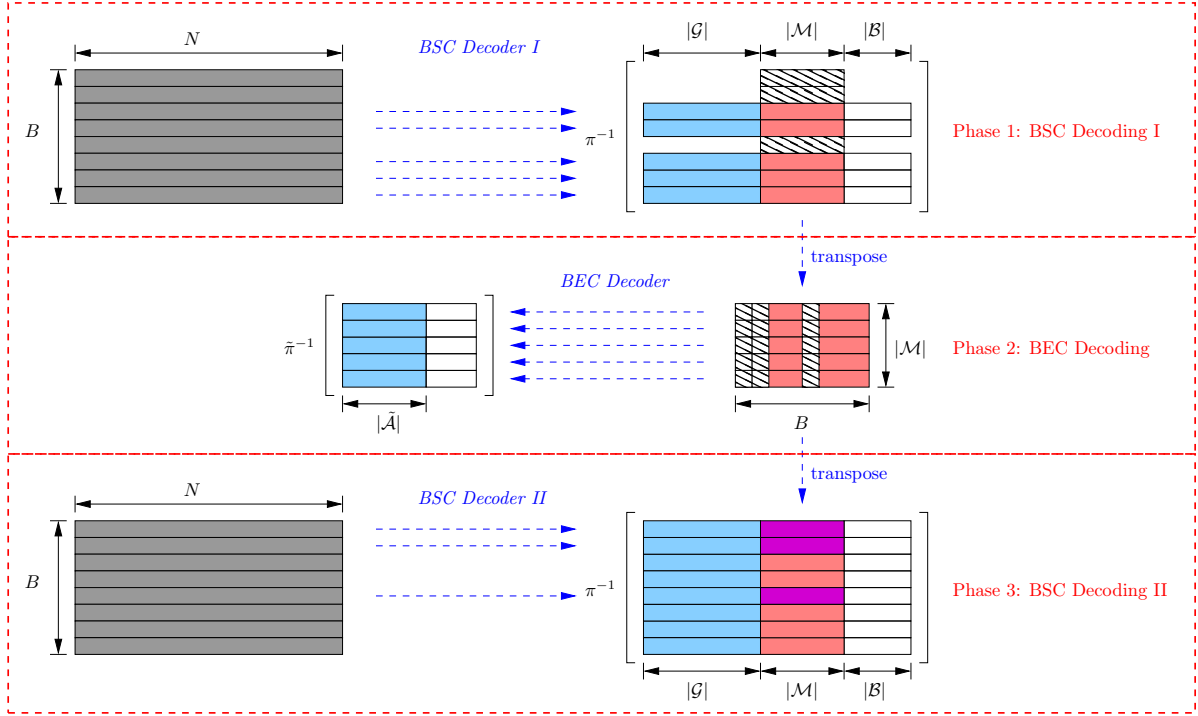}
 \caption{{\bf Illustration of proposed polar decoder for a fading binary symmetric channel with two states.} In Phase 1, decoder outputs all estimates using BSC SC decoders corresponding to the superior channel state. Selected columns are transposed and delivered as inputs to next phase, by adding all-erasures rows for blocks with the degraded channel state. In Phase 2, the decoder continues to use BEC SC decoders to decode all the blockwise information bits, and to recover all erased bits in shade. In Phase 3, the BSC SC decoders corresponding to the degraded channel state are utilized to decode the remaining information bits, by taking values of frozen bits in set $\mathcal{M}$ as the decoded results from the previous phase. $\tilde{\pi}$ and $\pi$ are column reordering permutations with respect to BEC and BSC, correspondingly.
}
\label{fig:Decoder}
\end{figure*}

The whole decoding process for fading binary symmetric channel is illustrated in Fig.~\ref{fig:Decoder}.

\subsection{Performance}

Here, we summarize the performance of the proposed polar coding scheme. Intuitively, by using BSC SC decoders corresponding to the superior channel state, the output from Phase 1 successfully recovers all information bits, because the size of information set is equal to the size of $\mathcal{G}$. Then, for decoding at Phase 2, the input vector $\hat{\tilde{u}}^{(k)}$ can be considered as a $q_1$-fraction erased polar codeword, hence, BEC SC decoder can decode all information bits in $v^{(k)}$ correctly for all $k\in\{1,\ldots,|\mathcal{M}|\}$, and recover the erased entries correctly as well. Finally, in Phase 3 of decoding, the bits in $\mathcal{M}$ have the correct frozen values, and by adopting BSC SC decoders corresponding to the degraded channel state, all the remaining information bits can be decoded correctly.

Therefore, as long as the designed rates of polar codes do not exceed the corresponding channel capacities, all information bits in our proposed polar coding scheme are reliably decodable. Hence, we have the following theorem.

%More formally, we have the following theorem.
\begin{theorem}
The proposed polar coding scheme achieves any rate $R<C_{\text{CSI-D}}$, for sufficiently large $N$ and $B$, and the decoding error probability scales as $\max\{O(B2^{-N^{\beta}}),O(N2^{-B^{\beta}})\}$ with $\beta<1/2$. Moreover, the complexity of the encoding and decoding processes are both given by $O(NB\log(NB))$, where $N$ is the block length and $B$ is the number of blocks.
\end{theorem}
\begin{proof}
The achievable rate (corresponding to the transmission of information bits in $v^{(k)}$ and $u^{(b)}$) is given by
\begin{align}
R   &=\frac{1}{NB}\Big\{|\mathcal{M}||\tilde{\mathcal{A}}|+B|\mathcal{G}|\Big\}\nonumber\\
    &=[H(p_1)-H(p_2)][1-q_1-\epsilon]+[1-H(p_1)-\epsilon]\nonumber\\
    &=q_1[1-H(p_1)]+q_2[1-H(p_2)]-\delta(\epsilon),\nonumber
%    &=C_{\text{SI-D}}-\delta(\epsilon),\nonumber
\end{align}
where we have used (\ref{fun:size_G}), (\ref{fun:size_M}), (\ref{fun:size_A}), and
\begin{align}
\delta(\epsilon)\triangleq\epsilon[1+H(p_1)-H(p_2)]\to 0, \text{ as }\epsilon\to 0. \nonumber
\end{align}

The proof for error exponent is obtained by utilizing error bound from polar coding. In Phase 1 and 3 of decoding, the error probability of recovering $u^{(b)}$ correctly for each $b\in\{1,\ldots,B\}$ is given by
$P_{1,e}^{(b)}=O(2^{-N^{\beta}})$. Similarly, in decoding Phase 2, the error probability of recovering $v^{(k)}$ correctly for each $k\in\{1,\ldots,|\mathcal{M}|\}$ is given by
$P_{2,e}^{(k)}=O(2^{-B^{\beta}})$.
Hence, by union bound, the total decoding error probability is upper bounded by
\begin{align}
P_{e}\leq\sum_{b=1}^{B}P_{1,e}^{(b)}+\sum_{k=1}^{|\mathcal{M}|}P_{2,e}^{(k)}=O(B2^{-N^{\beta}})+O(N2^{-B^{\beta}}),\nonumber
\end{align}
as $N$ and $B$ tend to infinity. Therefore, $P_e$ vanishes if $B=o(2^{N^{\beta}})$ and $N=o(2^{B^{\beta}})$.

Finally, since we have $|\mathcal{M}|$ number of $B$-length polar codes as well as $B$ number of $N$-length polar codes utilized, the overall complexity of the coding scheme for both encoding and decoding is given by
\begin{equation}
|\mathcal{M}|\cdot O(B\log B)+B\cdot O(N\log N)=O(NB\log (NB)).\nonumber
\end{equation}
\end{proof}

This theorem shows that our proposed polar coding scheme achieves the capacity of fading BSC with low encoding and decoding complexity. In addition, the error scaling performance, which is inherited from polar codes, implies that long coherence intervals as well as large number of blocks are required for this coding scheme to make the error probability arbitrarily small.

\subsection{Generalization}

Here, we generalize the polar coding scheme to fading binary symmetric channel with arbitrary finite number of states. Consider $S$ number of BSCs, each with a different transition probability. Without loss of generality, consider $W_1\triangleq \text{BSC}(p_1),\ldots,W_S\triangleq\text{BSC}(p_S)$, with $p_1\geq p_2\geq\cdots\geq p_S$. Then, a fading BSC with $S$ fading states is modeled as the channel being $W_s$ with probability $q_s$ for a given fading block, where $\sum\limits_{s=1}^Sq_s=1$. The polarization of a fading BSC with $S$ fading states is illustrated in Fig.~\ref{fig:Polarization_S}, where the reconstructed channel indices are divided into $S+1$ sets after permutation $\pi$. In addition to $\mathcal{G}$ and $\mathcal{B}$, there exist $S-1$ middle sets $\mathcal{M}_1$, \ldots, $\mathcal{M}_{S-1}$ in this case. For each channel index in set $\mathcal{M}_s$, channels having statistics being one of $W_1,\ldots,W_s$ are polarized to be purely noisy and the remaining ones are noiseless. Therefore, for channel indices belonging to $\mathcal{M}_s$, we consider modeling them as BECs with erasure probability given by
$$e_s\triangleq\sum_{t=1}^s q_t,\quad 1\leq s\leq S-1.$$
Based on this, we have
\begin{align}
&|\mathcal{G}|=|\mathcal{A}_1|=N[1-H(p_1)-\epsilon],\nonumber\\
&|\mathcal{M}_s|=N\left[H(p_s)-H(p_{s+1})\right],\quad 1\leq s\leq S-1,\nonumber\\
&|\mathcal{B}|=N-|\mathcal{A}_S|=N[H(p_{S})+\epsilon].\nonumber
\end{align}
\begin{figure}[t]
 \centering
 \includegraphics[width=0.8\columnwidth]{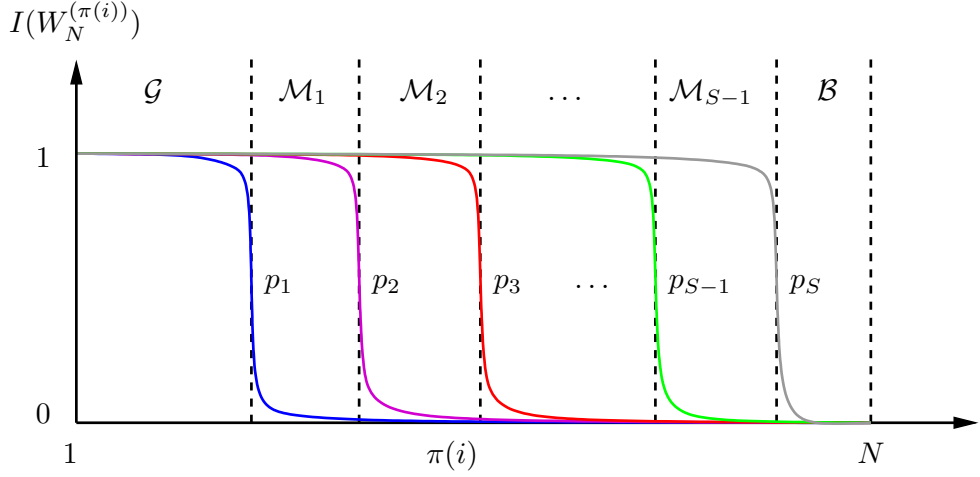}
 \caption{{\bf Illustration of polarization for a fading binary symmetric channel with $S$ channel states.} Besides $\mathcal{G}$ and $\mathcal{B}$, there are $S-1$ middle sets, denoted as $\mathcal{M}_1,\ldots,\mathcal{M}_{S-1}$.
}
\label{fig:Polarization_S}
\end{figure}

Here, the polarization result is similar to the case of two fading states, and we utilize a similar hierarchical coding scheme. In Phase 1 of encoding, transmitter generates $S-1$ sets of polar codes, where each one is a $G_B$-coset code with parameter $(B,|\tilde{\mathcal{A}}_s|,\tilde{\mathcal{A}}_s,0)$ with respect to BEC$(e_s)$, and all the encoded codewords are embedded into messages of Phase 2 in order. Then, in Phase 2 of encoding, we use BSC polar encoders with information set $\mathcal{G}$ to generate the final codeword with length $NB$. At the receiver end, we need $2S-1$ number of phases. Phase 1 utilizes the BSC$(p_S)$ SC decoders to decode blocks with respective to the best channel state (state $S$ in this case). Consider all decoded bits in $\mathcal{M}_{S-1}$, as well as adding erasures to undecoded blocks, we could decode all erased bits by using BEC$(e_{S-1})$ SC decoders in Phase 2. Then, using the decoded information as frozen values for blocks with respective to state $S-1$, BSC$(p_{S-1})$ SC decoders are adopted in Phase 3 to decode information bits in the blocks corresponding to channel state $S-1$. Recursively, all information bits for both BSC encoding and BEC encoding could be reliably decoded, as long as the designed rates of polar codes do not exceed corresponding channel capacities. Hence, by adopting this hierarchical polar coding scheme, the achievable rate is given by
\begin{align}
R   &=\frac{1}{NB}\Bigg\{B|\mathcal{G}|+\sum_{s=1}^{S-1}|\mathcal{M}_s||\tilde{\mathcal{A}}_s|\Bigg\}\nonumber\\
    &=[1-H(p_1)-\epsilon]+\sum_{s=1}^{S-1}[H(p_s)-H(p_{s+1})](1-e_s-\epsilon)\nonumber\\
    &=\sum_{s=1}^Sq_s[1-H(p_s)]-\delta'(\epsilon),\nonumber
\end{align}
where $\delta'(\epsilon)\triangleq\epsilon[1+H(p_1)-H(p_S)]$ tends to $0$ as $\epsilon\to0$. Thus, to this end, the proposed polar coding scheme achieves the capacity of channel, and the encoding and decoding complexities are both given by
\begin{equation}
\sum_{s=1}^{S-1}|\mathcal{M}_s|\cdot O(B\log B)+B\cdot O(N\log N)=O(NB\log (NB)),\nonumber
\end{equation}
which is independent to the value of $S$ as $\sum\limits_{s=1}^{S-1}|\mathcal{M}_s|\leq N$. For the same reason, the decoding error bound also remains the same as the case of only two fading states. Thus, our proposed polar coding scheme achieves the capacity of fading binary symmetric channel with arbitrary finite number of fading states, and the encoding and decoding complexity are both guaranteed to be tractable in practice.

Noting again the relevancy of this scenario to the fading AWGN channels, we consider another fading channel model with analog noise statistics in the next section, where the polar coding scheme proposed above is utilized.

%%%%%%%%%%%%%%%%%%%%%%%%%%%%%%%%%%%%%%%%%%%%%%%%%%%%%%%%%%%%%%%%%%%%%%%%%%%%%%
%%%%%%%%%%%%%%%%%%%%%%%%%%%%%%%%%%%%%%%%%%%%%%%%%%%%%%%%%%%%%%%%%%%%%%%%%%%%%%

\section{Polar Coding for Fading Additive Exponential Noise Channel}

\subsection{System Model}

We consider the fading additive exponential noise (AEN) channel given by
\begin{equation}\label{eq:AENChannel}
Y_{b,i}=X_{b,i}+Z_{b,i},\quad b=1,\ldots,B,\quad i=1,\ldots,N,
\end{equation}
where $X_{b,i}$ is channel input and restricted to be positive and with mean $E_X$; $N$ is block length; and $B$ is the number of blocks. In this model, $Z_{b,i}$ are assumed to be identically distributed within a block and follow an ergodic i.i.d. fading process over blocks. That is,
if we consider a fading AEN channel with $S$ states, then, with probability $q_s$ channel noise $Z_{b,i}$ is distributed as an exponential random variable with parameter $E_{Z_s}$ for a given $b$ and all $i\in\{1,\ldots,N\}$, i.e.,
\begin{equation}
f_{Z_{b,i}}(z) = \frac{1}{E_{Z_s}} e^{-\frac{z}{E_{Z_s}}},\quad z\geq 0,
\end{equation}
where $1\leq s\leq S$ and $\sum\limits_{s=1}^S q_s=1$.

We first state the following upper bound on the ergodic channel capacity in the high SNR regime.
\begin{lemma}\label{lem:Capacity_Fading_AEN}
The ergodic capacity of a fading AEN channel, with channel state information known at the decoder, is upper bounded as follows.
\begin{equation}
\lim\limits_{E_X\to\infty} C_{\textrm{CSI-D}}\leq \sum_{s=1}^S q_s\left[\log \left(1+\frac{E_X}{E_{Z_s}}\right)\right]\label{fun:fading_capacity}
\end{equation}
\end{lemma}
\begin{proof}
Denote the channel state as a random variable $G$, which is discrete on set $\{1,2,\ldots,S\}$. If the channel state information is known at the decoder, then we have
\begin{align}
&\lim\limits_{E_X\to\infty} C_{\textrm{CSI-D}}\nonumber\\
&\quad\stackrel{(a)}{\leq} \lim\limits_{E_X\to\infty} C_{\textrm{CSI-ED}} \nonumber\\
&\quad\stackrel{(b)}{=}\lim\limits_{E_X\to\infty} \max_{\mathbb{E}[X]\leq E_X}I(X;Y|G)\nonumber\\
                &\quad=\lim\limits_{E_X\to\infty} \max_{\mathbb{E}[X]\leq E_X}h(Y|G)-h(Y|G,X)\nonumber\\
                &\quad=\lim\limits_{E_X\to\infty} \max_{X_s: \sum\limits_s q_s\mathbb{E}[X_s] \leq E_X}\sum_{s=1}^S q_s[h(X_s+Z_s)-h(Z_s)]\nonumber\\
		&\quad\stackrel{(c)}{=}\lim\limits_{E_X\to\infty} \max_{E_{X_s}: \sum\limits_s q_sE_{X_s} \leq E_X} \sum_{s=1}^S q_s \left[\log\left( 1+\frac{E_{X_s}}{E_{Z_s}} \right)\right]\nonumber\\
                &\quad\stackrel{(d)}{=}\sum_{s=1}^S q_s \left[\log\left( 1+\frac{E_X}{E_{Z_s}} \right)\right],\nonumber
\end{align}
where $(a)$ is due to upper bounding the channel capacity with the case where encoder also has CSI and adapts its coding according to the channel states; $(b)$ is the ergodic capacity of the channel where both encoder and decoder has CSI (see, e.g., \cite[pages 203-209]{Tse:Wireless05}); and $(c)$ holds as exponential distribution maximizes the differential entropy on positive support with a mean constraint \cite{Verdu:Exponential96} \cite[page 412]{Cover:IT1991}. Here, we choose $X_s$ to be a weighted sum of an exponential distribution with mean $E_{X_s}+E_{Z_s}$ and a delta function in order to make the output $X_s+Z_s$ to be exponentially distributed random variable. That is, the pdf of $X_s$ is given by
\begin{align}
f_{X_s}(x)= \frac{E_{X_s}}{E_{X_s}+E_{Z_s}} \frac{e^{-x/(E_{X_s}+E_{Z_s})}}{E_{X_s}+E_{Z_s}}u(x)+ \frac{E_{Z_s}}{E_{X_s}+E_{Z_s}} \delta(x),\label{eq:optinput}
\end{align}
where $\delta(x)=1$ if $x=0$, and $\delta(x)=0$ if $x\neq 0$; $u(x)=1$ if $x\geq 0$, and $u(x)=0$ if $x<0$.
Finally, $(d)$ follows by taking the limit.
\end{proof}

In the following, we show that our proposed polar coding scheme achieves the upper bound above in the high SNR regime.

\begin{remark} Note that the capacity of the fading AEN channel with CSI-D approaches to the bound above in the high SNR regime. (For example, our coding scheme, as shown below, provides one such achievable rate.) This observation is similar to the Gaussian counterpart~\cite[pages 203-209]{Tse:Wireless05}, where in the high SNR regime, the performance obtained from a waterfilling strategy - the optimal solution for the case where encoder can adapt its power based on the channel state, i.e., CSI-ED - approaches to the performance of utilizing the same power allocation for each fading channel.
\end{remark}

\begin{remark} The model above assumes a mean constraint on the channel input where the average is over channel blocks and channel states. If the mean constraint is per block (abbreviated as MPB - Mean Per fading Block - in the following), i.e., $E[X_{b,i}]\leq E_X$ for each fading block $b$, then by following steps similar to the ones above, we have
$$C_{\textrm{CSI-D, MPB}} \leq C_{\textrm{CSI-ED, MPB}} = \sum_{s=1}^S q_s \left[\log\left( 1+\frac{E_X}{E_{Z_s}} \right)\right].$$
\end{remark}

\subsection{Binary Expansion of Exponential Distribution}

In \cite{Ozan:Expansion12}, expansion coding scheme is proposed for a static AEN channel, where the channel is expanded by the decomposition property of exponential random variable. Here, a similar scheme is adopted for the fading AEN channel. We first start with the following lemma, providing the theoretical basis for expansion coding.
\begin{lemma}\label{lem:EXP_EXP}
Let $A_l$'s be independent Bernoulli random variables with parameters
given by $a_l$, i.e., $\text{Pr}\{A_l=1\}\triangleq a_l$,
and consider the random variable defined by
\begin{equation}
A\triangleq \sum\limits_{l=-\infty}^{\infty} 2^l A_l.
\end{equation}
Then, the random variable $A$ is exponentially distributed
with mean $\lambda^{-1}$, i.e., its pdf is given by
\begin{equation}
f_A(a) = \lambda e^{-\lambda a}, \quad a\geq 0,
\end{equation}
if and only if the choice of $a_l$ is given by
\begin{equation}
a_l = \frac{1}{1+e^{\lambda 2^l}}. \label{fun:q_l}
\end{equation}
\end{lemma}
This lemma reveals that one can reconstruct an exponential random variable from a set of independent Bernoulli random variables perfectly.
The proof is given in \cite{Ozan:Expansion12}, and a set of typical numerical values of $a_l$s for $\lambda=1$ is shown in Fig.~\ref{fig:Exp_Prob}. It is evident that $a_l$ approaches $0$
for what we refer to as the ``higher'' levels and approaches $0.5$ for the ``lower'' levels. Hence, the primary non-trivial levels meaningful for coding are the ``middle'' ones. This observation provides the basis for truncating the number of levels to a finite value without a significant loss in performance.
\begin{figure}[t]
 \centering
 \includegraphics[width=0.8\columnwidth]{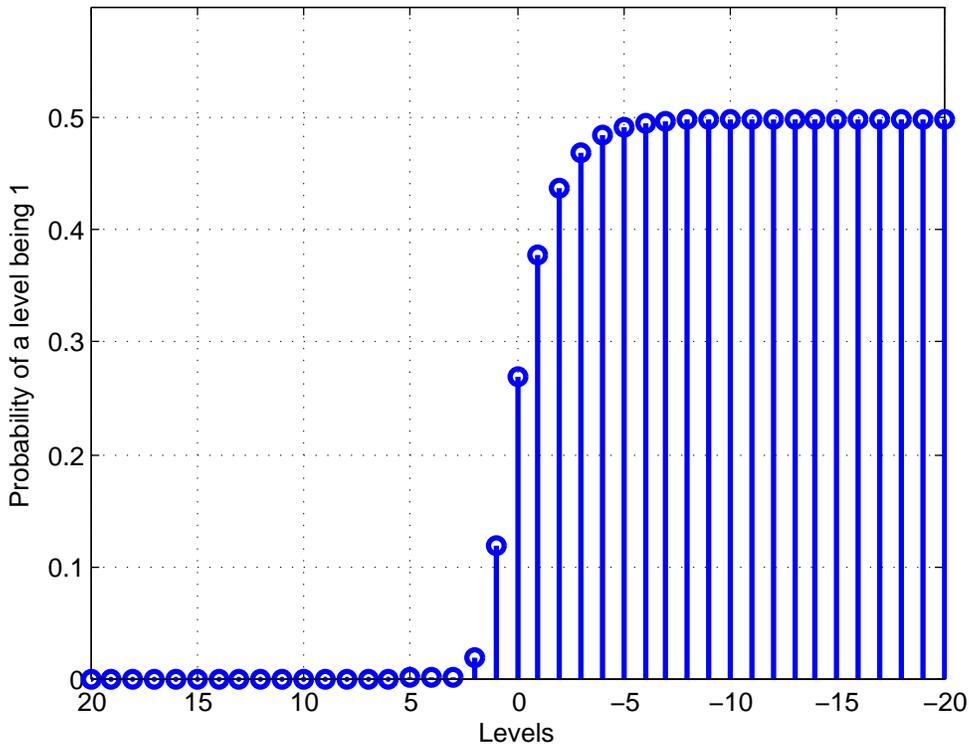}
 \caption{{\bf Numerical result.} A set of $a_l$s, the parameters from expanded levels, are shown, where the target random variable expanded from is exponentially distributed with $\lambda=1$.}
\label{fig:Exp_Prob}
\end{figure}

\subsection{Expansion Coding}

\begin{figure}[t]
 \centering
 \includegraphics[width=0.8\columnwidth]{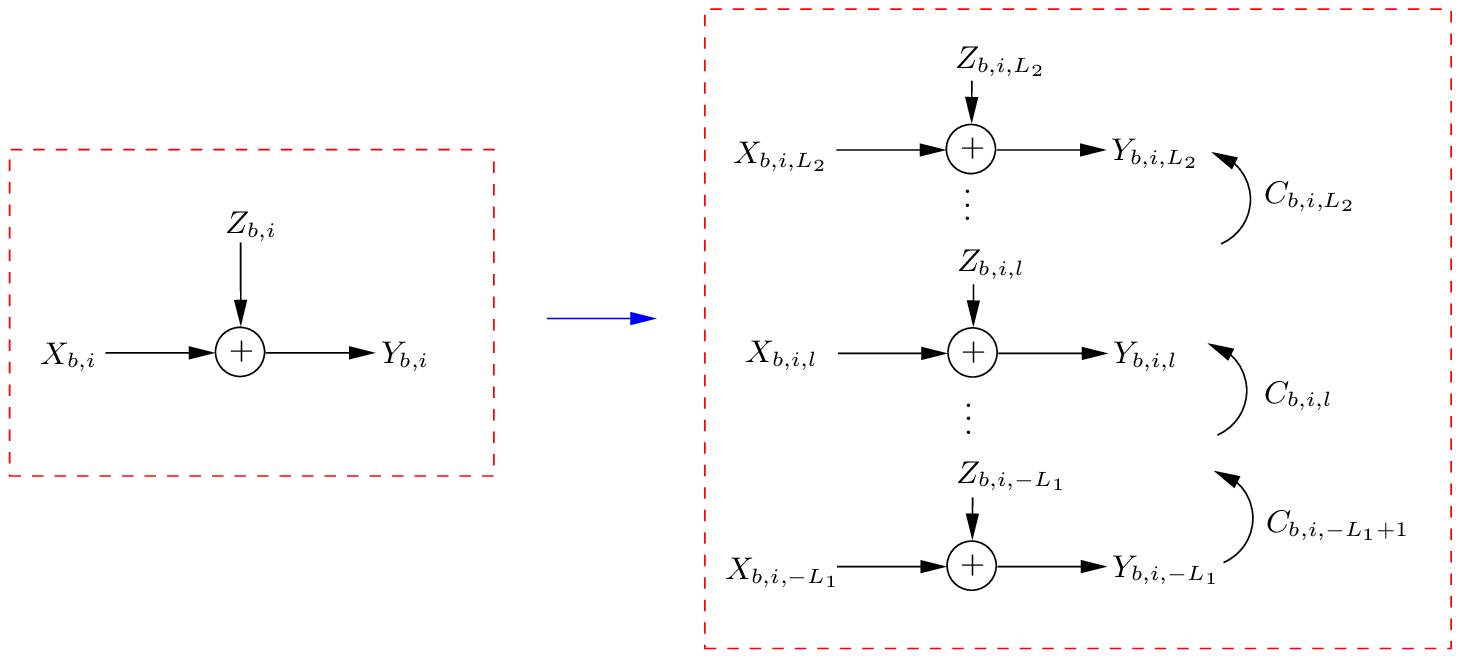}
 \caption{{\bf Illustration of expansion channel coding scheme.} An analog channel is expanded into a set of binary symmetric channels, where carries are considered between neighboring levels.
}
\label{fig:Expansion}
\end{figure}

We consider the binary expansion of channel noise
\begin{align}
\hat{Z}_{b,i} \triangleq \sum\limits_{l=-L_1}^{L_2} 2^{l} Z_{b,i,l},
\end{align}
where $Z_{b,i,l}$ is a discrete random variable taking value in $\{0,1\}$. However, the distribution of $Z_{b,i,l}$ depends on the fading state. More precisely, if the noise for a fading block $b$ is exponential with parameter $E_{Z_s}$, then $Z_{b,i,l}$ is a Bernoulli random variable with parameter
\begin{equation}
\tilde{p}_{l,s}\triangleq \text{Pr}\{Z_{b,i,l}=1\}=\frac{1}{1+e^{2^l/E_{Z_s}}}.
\end{equation}
Then, by Lemma \ref{lem:EXP_EXP}, $\hat{Z}_{b,i}\overset{d.}{\to} Z_{b,i}$ as $L_1$ and $L_2$ tend to infinity. In this sense, we approximate the original exponential noise perfectly by a set of discrete noises.

Similarly, we also expand channel input and output as follows,
\begin{align}
\hat{Y}_{b,i} \triangleq \sum\limits_{l=-L_1}^{L_2} 2^{l} Y_{b,i,l}= \sum\limits_{l=-L_1}^{L_2} 2^{l} (X_{b,i,l}+Z_{b,i,l}),
\end{align}
where $X_{b,i,l}$ is also a Bernoulli random variable with parameter $p_l\triangleq \text{Pr}\{X_{b,i,l}=1\}$.
At this point, we model the expanded channels as
\begin{equation}
Y_{b,i,l}=X_{b,i,l}+Z_{b,i,l},\quad l=-L_1,\ldots,L_2.
\end{equation}
Note that the summation is a real sum here, and hence, the channel is not a fading BSC for a given block.
If we replace the real sum by modulo-$2$ sum, then, at level $l$, any capacity achieving code for fading BSC, for example the one constructed in Section~\ref{sec:FadingBSC}, can be utilized over this channel with optimal
input probability distribution. In addition, instead of directly using the capacity achieving
code, one can consider its combination with the method of Gallager~\cite[pages 74-76]{Gallager:Information68} \cite{Korada:Thesis09} to
achieve a rate corresponding to the one obtained by the mutual
information $I(X_{b,l};Y_{b,l})$ evaluated with a desired input distribution on $X_{b,l}$.
Hence, a desired rate (evaluation of $I(X_{b,l};Y_{b,l})$ with some distribution on $X_{b,l}$) at level $l$ and fading block $b$ can be achieved.

However, due to the real sum in the original channel above, carries exist between neighboring levels (see Fig.~\ref{fig:Expansion}), which further implies that coding over levels are not independent. Hence, we do not have independent parallel channels to start with.
Every level, except for the lowest one, is impacted by carries accumulated from lower levels. \cite{Ozan:Expansion12} proposed a scheme to get rid of this issue, where carries are decoded from the lowest level to the highest one. In this way, channels over levels are transformed to
behave independently (assuming reliable decoding of each carry w.h.p.), and the total achievable rate is the summation of individual achievable rates over all levels.

Using this technique (to essentially remove carries) as suggested in \cite{Ozan:Expansion12}, each level could be modeled as a fading BSC. Thus, expansion coding reduces the problem of coding over a fading exponential noise channel into a set of simpler subproblems, coding over fading BSCs. By adopting capacity achieving polar coding scheme proposed in Section \ref{sec:FadingBSC}, we have the following achievable rate result for these channels.

\begin{theorem}\label{thm:AEN_optimization}
By decoding carries in expansion coding, and adopting hierarchical polar coding scheme for fading BSC in each expanded level, the proposed scheme achieves the rate given by
\begin{equation}
R=\sum_{l=-L_1}^{L_2}\sum_{s=1}^Sq_s[H(p_l\otimes \tilde{p}_{l,s})-H(\tilde{p}_{l,s})],\label{equ:achievable_rate}
\end{equation}
for any $L_1,L_2>0$, where $p_l\in[0,0.5]$ is chosen to satisfy
\begin{equation}
\sum_{l=-L_1}^{L_2}2^l p_l\leq E_X.
\end{equation}
\end{theorem}

We note the followings. First, the achievable scheme we utilize satisfies the mean constraint on the channel input for each block, i.e., averaged over channel uses, $\lim\limits_{N\to\infty} \frac{1}{N} \sum\limits_{i=1}^N X_{i,b} \leq E_X$ for each block $b$. (This implies satisfying power constraint averaged over the blocks as well.) Secondly, the maximum rate from our coding scheme could be considered as an optimization problem over finite number of parameters $p_l$, $-L_1\leq l\leq L_2$. However, it is not clear how to solve this non-convex problem. Here, instead of searching for an optimal solution, we shift our focus to finding a sub-optimal choice of $p_l$ such that the achievable rate is close the optimal one in the high SNR regime. From \eqref{eq:optinput}, we observe that the optimal input distribution for the case with the CSI at the transmitter could be approximated with an exponential with parameter $E_{X_s}$ as $\textrm{SNR}=E_{X_s}/E_{Z_s}$ gets large. As we do not have CSI at the transmitter in our model, we consider choosing the same energy level, $E_{X}$, for each fading block. Noting again that the optimal input distribution is unknown for our fading model, the high SNR observation inspires us to choose $p_l$ as
\begin{equation}
p_l=\frac{1}{1+e^{2^l/E_X}}. \label{equ:AEN_p_l}
\end{equation}

The following theorem gives the main result of our polar coding scheme over fading AEN channel. (We accompany the proof of this theorem with Fig.~\ref{fig:AEN_Rate_Level} in order to provide not only the details of the proof but also the intuition on how expansion approach is helpful.)

\begin{theorem}\label{thm:AENrate}
For any positive constant $\epsilon<1$, if
\begin{itemize}
\item $L_1\geq -\log \epsilon-\min\limits_s \log E_{Z_s}$;
\item $L_2\geq -\log \epsilon+\log E_X$;
\item $\min\limits_s \textrm{SNR}_s\geq 1/\epsilon$, where $\text{SNR}_s\triangleq E_X/E_{Z_s}$,
\end{itemize}
then by decoding carries and adopting hierarchical polar codes at each fading BSC after expansion, the achievable rate $R$ given by \eqref{equ:achievable_rate}, with a choice of $p_l$ as \eqref{equ:AEN_p_l}, satisfies
$$R\geq \sum_{s=1}^{S}q_s\left[\log\left(1+\frac{E_X}{E_{Z_s}}\right)\right]-5\log e\cdot\epsilon.$$
\end{theorem}

\begin{proof}
We first state a bound for the entropy of channel noise with mean $E_{Z_s}$ at level $l$
\begin{align}
&H(\tilde{p}_{l,s})\leq3\log e\cdot 2^{-l+\eta_s}   \;\quad \text{for }l>\eta_s,\label{equ:bound1}\\
&H(\tilde{p}_{l,s})\geq1-\log e\cdot 2^{l-\eta_s}      \quad \text{for }l\leq \eta_s,\label{equ:bound2}
\end{align}
where $\eta_s \triangleq\log E_{Z_s}$. The proofs of these bounds for the case of $E_{Z_s}=1$ are detailed in Lemma~4 of \cite{Ozan:Expansion12}. Here, we obtain these bounds by following the same steps given there. (Details are omitted for brevity.)

Now, if we denote $\xi\triangleq\log E_X$, then comparing the definitions of $p_l$ and $\tilde{p}_{l,s}$, we get
\begin{equation}
p_l=\frac{1}{1+e^{2^{l}/E_X}}=\tilde{p}_{l+\eta_s-\xi,s}.\label{equ:pqequation}
\end{equation}
Based on these observations, we have
\begin{align}
    \sum_{l=-L_1}^{L_2}&[H(p_l\otimes \tilde{p}_{l,s})-H(\tilde{p}_{l,s})]\nonumber\\
\overset{(a)}{\geq}&\sum_{l=-L_1}^{L_2}[H(p_l)-H(\tilde{p}_{l,s})]\nonumber\\
\overset{(b)}{=}   &\sum_{l=-L_1}^{L_2}[H(\tilde{p}_{l+\eta_s-\xi,s})-H(\tilde{p}_{l,s})]\nonumber\\
=   &\sum_{l=-L_1+\eta_s-\xi}^{L_2+\eta_s-\xi}H(\tilde{p}_{l,s})-\sum_{l=-L_1}^{L_2}H(\tilde{p}_{l,s})\nonumber\\
=   &\sum_{l=-L_1+\eta_s-\xi}^{-L_1-1}H(\tilde{p}_{l,s})-\sum_{l=L_2+\eta_s-\xi+1}^{L_2}H(\tilde{p}_{l,s})\nonumber\\
\overset{(c)}{\geq}&\sum_{l=-L_1+\eta_s-\xi}^{-L_1-1}\left[1-\log e\cdot 2^{l-\eta_s}\right]-\sum_{l=L_2+\eta_s-\xi+1}^{L_2}3\log e\cdot 2^{-l+\eta_s}\nonumber\\
\overset{(d)}{\geq}&\xi-\eta_s-\log e\cdot 2^{-L_1-\eta_s}-3\log e\cdot 2^{-L_2+\xi}\nonumber\\
\overset{(e)}{\geq}&\log\left(\frac{E_X}{E_{Z_s}}\right)-\log e\cdot \epsilon-3\log e\cdot \epsilon\nonumber\\
\overset{(f)}{\geq}&\log\left(1+\frac{E_X}{E_{Z_s}}\right)-\log e\cdot\frac{E_{Z_s}}{E_X}-\log e\cdot \epsilon-3\log e\cdot\epsilon\nonumber\\
\overset{(g)}{\geq}&\log\left(1+\frac{E_X}{E_{Z_s}}\right)-5\log e \cdot \epsilon,\label{eq:AEN_proof}
\end{align}
where
\begin{itemize}
\item[$(a)$] is due to $p_l\otimes \tilde{p}_{l,s}\triangleq p_l(1-\tilde{p}_{l,s})+\tilde{p}_{l,s}(1-p_l)\geq p_l$, and then due to the fact that entropy function is increasing on $[0,0.5]$ (and, we have $p_l\otimes \tilde{p}_{l,s}\leq 0.5$);
\item[$(b)$] follows from equation \eqref{equ:pqequation};
\item[$(c)$] follows from bounds \eqref{equ:bound1} and \eqref{equ:bound2};
\item[$(d)$] follows as
$$\sum\limits_{l=-L_1+\eta_s-\xi}^{-L_1-1} 2^{l-\eta_s} \leq 2^{-L_1-\eta_s},$$
 and
$$\sum\limits_{l=L_2+\eta_s-\xi+1}^{L_2} 2^{-l+\eta_s} = \sum\limits^{-L_2+\xi-1}_{l=-L_2+\eta_s} 2^l \leq 2^{-L_2+\xi};$$
\item[$(e)$] follows from theorem assumptions that $L_1\geq -\log \epsilon-\min\limits_s \eta_s$, and $L_2\geq -\log \epsilon+\xi$;
\item[$(f)$] is due to the fact that $\log(1+E_X/E_{Z_s})-\log(E_X/E_{Z_s})=\log(1+E_{Z_s}/E_{X})\leq \log e \cdot E_{Z_s}/E_X$ (as $\ln(1+x)\leq x$ for any $x\geq 0$);
\item[$(g)$] is due to the assumption in theorem that $\min\limits_s \text{SNR}_s\geq 1/\epsilon$.
\end{itemize}
Then, using \eqref{eq:AEN_proof} in \eqref{equ:achievable_rate} of Theorem~\ref{thm:AEN_optimization}, we have
\begin{align}
R   &=\sum_{s=1}^{S}q_s\left\{\sum_{l=-L_1}^{L_2}[H(p_l\otimes \tilde{p}_{l,s})-H(\tilde{p}_{l,s})]\right\}\nonumber\\
    &\geq \sum_{s=1}^{S}q_s\left\{\log\left(1+\frac{E_X}{E_{Z_s}}\right)-5\log e \cdot \epsilon\right\}\nonumber\\
    &=\sum_{s=1}^{S}q_s\left[\log\left(1+\frac{E_X}{E_{Z_s}}\right)\right]-5\log e \cdot \epsilon. \nonumber
\end{align}
\end{proof}

\begin{remark}
We note that the proposed scheme achieves a rate
$$\sum_{s=1}^{S}q_s\left[\log\left(1+\frac{E_X}{E_{Z_s}}\right)\right]=C_{\text{CSI-ED, MPB}},$$
which is an upper bound on the capacity for the CSI-D case in the high SNR regime. (See Lemma~\ref{lem:Capacity_Fading_AEN}.)
Therefore, the proposed scheme achieves the capacity in the high SNR regime.
\end{remark}

\begin{figure}[t]
 \centering
 \includegraphics[width=0.8\columnwidth]{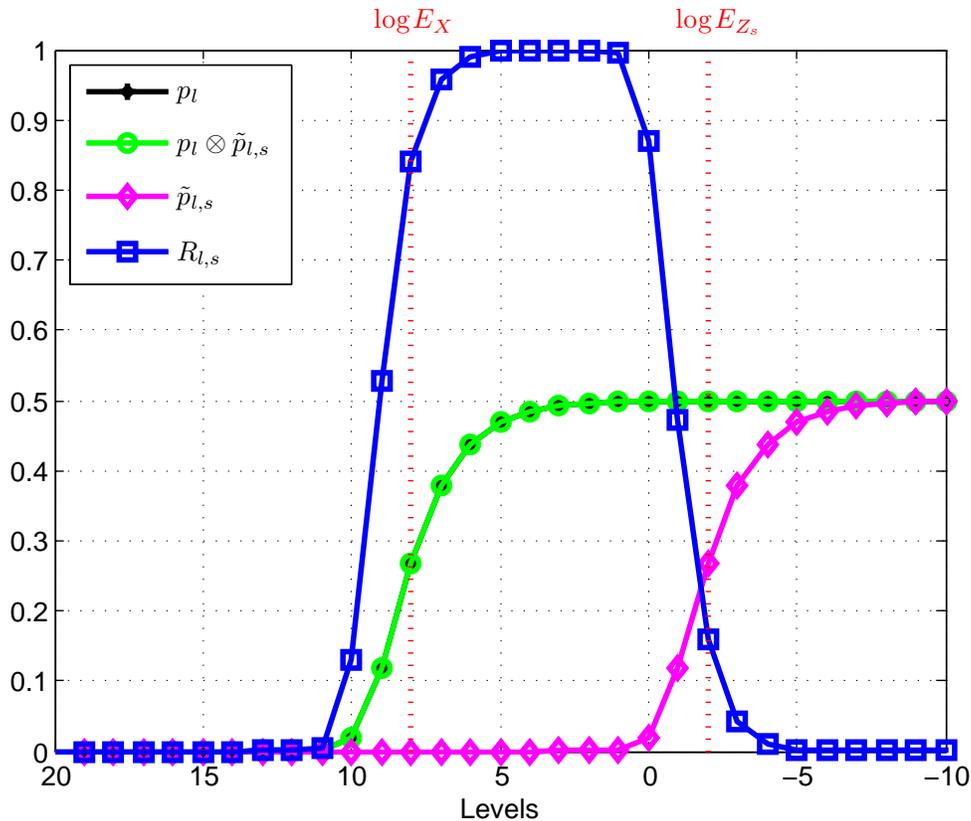}
 \caption{{\bf Illustration of signal, noise and rate at each level.} In this numerical result, we only concern about a single state $s$, and choose $E_X=2^8$, $E_{Z_s}=2^{-2}$. Note that, in this case, the curve of $p_l\otimes \tilde{p}_{l,s}$ (in black) almost coincides with the one of $p_l$ (in green), because of high SNR. Moreover, the signal $p_l$ (in green) is a left shifted version of noise $\tilde{p}_{l,s}$ (in purple) by $\log E_X-\log E_{Z_s}$ number of levels. Achievable rates at each level, $R_{l,s}\triangleq H(p_l\otimes \tilde{p}_{l,s})-H(\tilde{p}_{l,s})$, are represented in blue.
}
\label{fig:AEN_Rate_Level}
\end{figure}

Note that, for an exponential distribution with mean $1/\lambda$, its binary expansion result can be considered as the expansion of an exponential distribution with mean $1$ shifted by $\log (1/\lambda)$ number of levels. We show this phenomenon in Fig.~\ref{fig:AEN_Rate_Level}. To this end, Theorem~\ref{thm:AENrate} shows that in order to achieve the capacity of fading AEN channel, first, SNR should be large enough, and secondly, the number of expanded levels should also be large enough such that the highest level exceeds all the left shifted levels of expanded signal, and the lowest level exceeds the right shifted levels of expanded noises. Hence, in total, basically we need $\log \text{SNR}_{\max}$ ($\text{SNR}_{\max}\triangleq \max\limits_s E_X/E_{Z_s}$) number of levels to cover all ``non-trivial'' levels for coding, as well as extra $-2\log \epsilon$ number of levels to shoot for accuracy. At this point, the complexities of encoding and decoding are both given by $O\big((\log \text{SNR}_{\max}-2\log\epsilon)NB\log (NB)\big)$, where $O(NB\log (NB))$ is the complexity scale for fading BSC derived in the previous section.

\subsection{Numerical Results}

In this section, we analyze the rate obtained from Theorem~\ref{thm:AEN_optimization} with parameter $p_l$ chosen as \eqref{equ:AEN_p_l}. Numerical results are illustrated in Fig.~\ref{fig:AEN_Rate}, where we consider the case of two fading states. It is evident from the figure, and also from the theoretical analysis given in Theorem~\ref{thm:AENrate}, that our proposed polar coding scheme together with expansion coding achieves the upper bound on the channel capacity (Lemma~\ref{lem:Capacity_Fading_AEN}) in the high SNR regime. Therefore, the proposed coding scheme achieves the channel capacity for sufficiently large SNR.

We also note that the coding scheme does not perform well in the low SNR regime, which mainly results from two reasons. First, the upper bound we derived in Lemma~\ref{lem:Capacity_Fading_AEN}, which is the target rate in our coding scheme, is not tight in the low SNR regime. Secondly, our choice of $p_l$ only behaves as a good approximation for sufficiently high SNR, which limits the proposed scheme to be effective at the corresponding regime. However, as evident from the numerical results, for a fairly large set of SNR values the proposed scheme is quite effective. In addition, the upper bound curve is equal to $C_{\text{CSI-ED, MPB}}$, the capacity when the input mean constraint is imposed per block (instead of averaging over the blocks). Therefore, for the scenario of having input constraint per each fading block, the upper bound $C_{\text{CSI-D, MPB}} \leq C_{\text{CSI-ED, MPB}}$ holds at any SNR, and the only degradation in our coding scheme is due to the second point discussed above.

\begin{figure}[t]
 \centering
 \includegraphics[width=0.8\columnwidth]{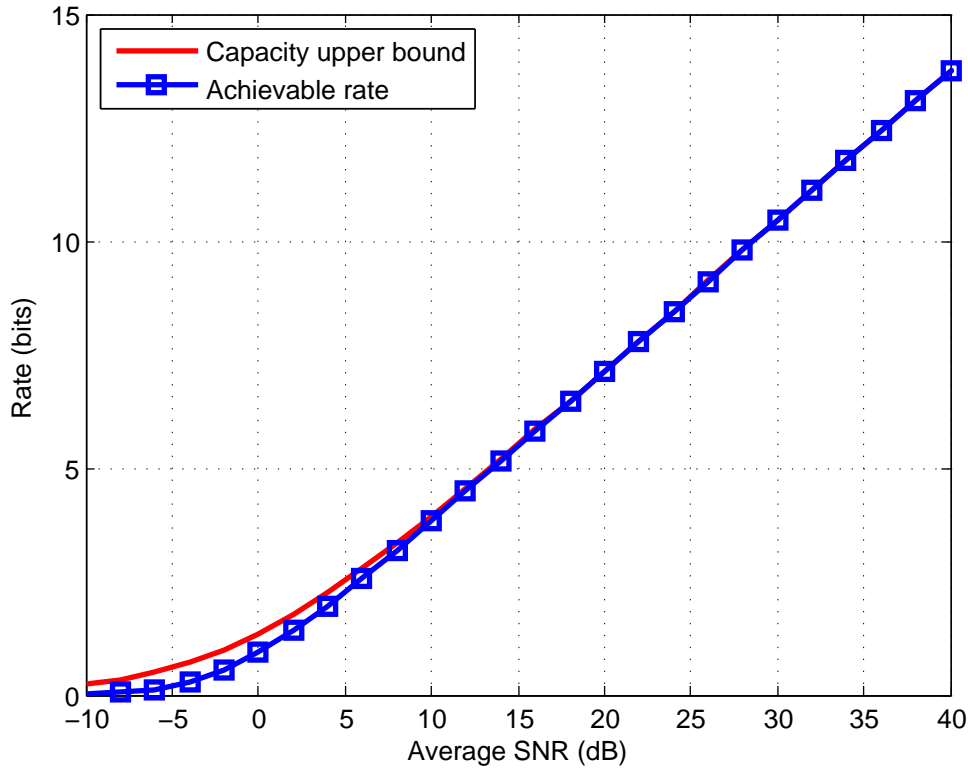}
 \caption{{\bf Numerical results.} The upper bound of ergodic capacity, $C_{\text{CSI-ED, MPB}}$, which is equal to $C_{\text{CSI-ED}}$ for sufficiently large SNR, is given by the red curve. The achievable rate is given by the blue curve. In this analysis, only two fading states are concerned, and the parameters are chosen as $E_{Z_1}=0.5$, $E_{Z_2}=3$, $q_1=0.8$, and $q_2=0.2$. Average SNR is defined as $E_X/(\sum\limits_{s=1}^S q_sE_{Z_s})$.
}
\label{fig:AEN_Rate}
\end{figure}

%%%%%%%%%%%%%%%%%%%%%%%%%%%%%%%%%%%%%%%%%%%%%%%%%%%%%%%%%%%%%%%%%%%%%%%%%%%%%%
%%%%%%%%%%%%%%%%%%%%%%%%%%%%%%%%%%%%%%%%%%%%%%%%%%%%%%%%%%%%%%%%%%%%%%%%%%%%%%
\section{Conclusion}

In this paper, polar coding schemes for fading binary symmetric channel (BSC) and fading additive exponential noise (AEN) channel are proposed. First, a hierarchical polar coding scheme is proposed for the fading BSC. This novel scheme, by exploiting an erasure decoding approach at the receiver, utilizes the polarization results of different BSCs. (These BSCs are defined over channel uses at a given fading block and over fading blocks at a given channel use index.) This novel polar coding technique is shown to be capacity achieving for fading BSC. Remarkably, the proposed scheme does not assume channel state information at the transmitter and fading BSC models the fading additive white Gaussian noise (AWGN) channel with a BPSK modulation. Therefore, our results are quite relevant to the practical channel models considered in wireless communications.

Towards utilizing the proposed techniques for encoding over another fading channel model, we focused on fading AEN channel. For this model, expansion coding \cite{Ozan:Expansion12} is adopted to convert the problem of coding over an analog channel into coding over discrete channels. By performing this expansion approach and making the resulting channels independent (via decoding the underlying carries), a fading AEN channel is decomposed into multiple independent fading BSCs (with a reliable decoding of the carries). By utilizing the hierarchical polar coding scheme for fading BSC, both theoretical proof and numerical results showed that the proposed approach achieves the capacity of this fading channel in the high SNR regime.

We remark that the advantages of polar codes in rate and complexity are both inherited in the proposed coding schemes. More precisely, as polar codes achieve channel capacity of BSC and BEC, our hierarchical utilization of polar codes also achieves the capacity of fading BSC, and this result is further utilized to guarantee that expansion coding scheme can achieve the capacity of fading AEN channel in the high SNR regime (with low complexity in all cases).

Although the discussion in this paper focuses only on fading BSC and fading AEN channel, the proposed coding scheme could be generalized to more general cases. For example, by utilizing non-binary polar codes, our polar coding scheme can be generalized to a fading non-binary discrete symmetric channel. This result can then be utilized for an AWGN fading channel with more constellation points, such as QPSK. Moreover, the expansion coding scheme can also be used for other analog channels that have noise statistics other than exponential, e.g., Gaussian. Here, even though these distributions may not be perfectly approximated by a set of independent discrete random variables, expansion coding scheme can still perform well, especially at high SNR.

Finally, we note that the proposed coding scheme requires long codeword lengths to make the error probability arbitrarily small. This requirement translates to requiring long coherence intervals and large number of fading blocks as our approach utilizes coding over both channel uses and fading blocks. (This is somewhat similar to the analyses in Shannon theory, where the guarantee of the coding is that the error probability vanishes as the block length gets large.) Therefore, our coding scheme fits to the fading channels with moderate/long coherence time and large number of fading blocks. Here, we comment on applicability of the proposed coding scheme in typical wireless systems. As reported in \cite[page 219]{sesia2009lte}, LTE systems operating at $1.8$GHz frequency with $20$MHz bandwidth typically have fading durations of $2.8\times10^5$ to $1.0\times10^7$ channel uses. In addition, WiFi systems operating at $5$GHz frequency with $20$MHz bandwidth typically have fading durations of $7.7\times 10^5$ to $1.8\times 10^7$ channel uses \cite[pages 98-99]{perahia2008next}. (Here, a mobile speed of $1$m/s is assumed for both systems.) Polar codes, on the other hand, typically have error rates around $10^{-6}$ when the blocklength is around $2^{10}$, and a smaller error probability is even possible, when the decoding is implemented with a better decoder \cite{Tse:Polar12}. For instance, instead of the classical SC decoder, a list decoder \cite{Tse:Polar12} can be utilized. Finally, besides long coherence intervals, another requirement for the proposed coding scheme is to have large number of fading blocks. This requirement can be satisfied in many practical scenarios at the expense of having large decoding delays. To summarize, for a given wireless system and a choice of one of the encoding/decoding strategies discussed above, the resulting error rate and its propagation in the proposed decoding algorithm should be studied further. We leave the analysis of such applications of the proposed techniques to a future work.

%%%%%%%%%%%%%%%%%%%%%%%%%%%%%%%%%%%%%%%%%%%%%%%%%%%%%%%%%%%%%%%%%%%%%%%%%%%%%%
%%%%%%%%%%%%%%%%%%%%%%%%%%%%%%%%%%%%%%%%%%%%%%%%%%%%%%%%%%%%%%%%%%%%%%%%%%%%%%
%\section*{Acknowledgment}
%
%The authors are thankful to reviewers and the Associate Editor for their valuable comments which significantly improved the paper.

%%%%%%%%%%%%%%%%%%%%%%%%%%%%%%%%%%%%%%%%%%%%%%%%%%%%%%%%%%%%%%%%%%%%%%%%%%%%%%
%%%%%%%%%%%%%%%%%%%%%%%%%%%%%%%%%%%%%%%%%%%%%%%%%%%%%%%%%%%%%%%%%%%%%%%%%%%%%%

\bibliographystyle{IEEEtran}

%%%%%%%%%%%%%%%%%%%%%%%%%%%%%%%%%%%%%%%%%%%%%%%%%%%%%%%%%%%%%%%%%%%%%%%%%%%%%%
%%%%%%%%%%%%%%%%%%%%%%%%%%%%%%%%%%%%%%%%%%%%%%%%%%%%%%%%%%%%%%%%%%%%%%%%%%%%%%

\end{document}